\documentclass[a4paper, 11pt, fleqn]{article}
\usepackage{amsmath, amssymb, amsthm}
\usepackage[margin=1in]{geometry}
\usepackage{xcolor}
\usepackage{url}
\usepackage{float}
\usepackage{array}
\usepackage{longtable}
\usepackage{comment}
\usepackage{tikz}
\usepackage{enumitem}
\usepackage{caption}
\usepackage[hypertexnames=false]{hyperref}
\usepackage{algorithm, algpseudocode}
\usepackage[capitalize,sort]{cleveref}

\hypersetup{
    colorlinks,
    linkcolor={red!50!black},
    citecolor={red!50!black},
    urlcolor={blue!50!black}
}

\newtheorem{theorem}{Theorem}

\newtheorem{lemma}[theorem]{Lemma}

\newtheorem{remark}[theorem]{Remark}

\crefname{claim}{Claim}{Claims}
\crefname{property}{Property}{Properties}
\crefname{transformation}{Transformation}{Transformations}
\algnewcommand{\LineComment}[1]{\State \textcolor{gray}{// #1}}
\newcolumntype{L}{>{$\displaystyle}l<{$}}

\newcommand*{\geomvec}[2]{(#1, #2)}
\newcommand*{\geomvecdim}[2]{(#1, #2)-dimensional}

\newcommand*{\vmgap}{Vector-Max-GAP}

\renewcommand*{\epsilon}{\varepsilon}
\newcommand*{\eps}{\epsilon}

\newcommand*{\floor}[1]{\left\lfloor #1 \right\rfloor}
\newcommand*{\ceil}[1]{\left\lceil #1 \right\rceil}
\newcommand*{\abs}[1]{\left\lvert #1 \right\rvert}

\newcommand*{\Th}{^{\textrm{th}}}

\DeclareMathOperator*{\argmin}{argmin}

\newcommand*{\WLoG}{W.l.o.g.}
\newcommand*{\wLoG}{w.l.o.g.}

\DeclareMathOperator*{\opt}{opt}
\newcommand*{\OPT}{\mathrm{OPT}}

\DeclareMathOperator{\poly}{poly}
\DeclareMathOperator{\round}{\mathtt{round}}

\usepackage{catchfile}

\ifx\input{"|kpsewhich -var-value RELEASE"}\empty
\def\rem#1{{\marginpar{\raggedright\scriptsize #1}}}
\newcommand*{\todo}[1]{\textcolor{violet}{TODO: #1}}
\newcommand*{\aricolor}{red}
\newcommand*{\eklcolor}{blue}
\newcommand*{\kvncolor}{green!50!black}

\newcommand*{\arir}[1]{\rem{\textcolor{\aricolor}{$\bullet$ #1}}}
\newcommand*{\eklr}[1]{\rem{\textcolor{\eklcolor}{$\bullet$ #1}}}
\newcommand*{\kvnr}[1]{\rem{\textcolor{\kvncolor}{$\bullet$ #1}}}
\else
\newcommand*{\todo}[1]{}

\newcommand{\arir}[1]{}
\newcommand{\eklr}[1]{}
\newcommand{\kvnr}[1]{}
\fi

\newcommand*{\epssmall}{\epsilon_{\mini}}
\newcommand*{\epslarge}{\epsilon_{\mammoth}}

\newcommand*{\lengthsym}{w}
\newcommand*{\heightsym}{h}
\newcommand*{\areasym}{a}
\newcommand*{\weightsym}{v}
\newcommand*{\profitsym}{p}

\newcommand*{\itemarea}[1]{\areasym(#1)}
\newcommand*{\itemlength}[1]{\lengthsym(#1)}
\newcommand*{\lengthinp}{\overrightarrow{\lengthsym}}
\newcommand*{\itemheight}[1]{\heightsym(#1)}
\newcommand*{\heightinp}{\overrightarrow{\heightsym}}
\newcommand*{\itemweight}[2]{\weightsym_{#1}(#2)} 
\newcommand*{\weightvec}[1]{\overrightarrow{\weightsym}(#1)} 
\newcommand*{\weightinp}{\overrightarrow{\weightsym}} 
\newcommand*{\itemprofit}[1]{\profitsym(#1)}
\newcommand*{\profitinp}{\overrightarrow{\profitsym}}

\newcommand*{\containerinstance}[3]{(#1, \lengthinp, \heightinp, \profitinp, \weightinp, #2, #3)}
\newcommand*{\vect}[1]{\protect\overrightarrow{#1}}
\newcommand*{\val}{\mathrm{val}}
\newcommand*{\Val}{\mathrm{VAL}}
\newcommand*{\vmg}{\textrm{Vector-Max-GAP}}

\newcommand*{\yaxis}{$y$-axis}
\newcommand*{\xaxis}{$x$-axis}

\newcommand*{\optgvks}{\operatorname{OPT}_{\mathrm{GVKS}}}

\newcommand*{\bigOhConst}[1]{O_{#1}(1)}
\newcommand*{\bigOh}[1]{O(#1)}
\newcommand*{\skewed}{\mathrm{skew}}
\newcommand*{\corr}{\mathrm{corr}}
\newcommand*{\cross}{\mathrm{cross}}
\newcommand*{\nice}{\mathrm{nice}}
\newcommand*{\bad}{\mathrm{bad}}
\newcommand*{\boxrect}{\mathrm{box}}
\newcommand*{\ra}{\mathrm{ra}}
\newcommand*{\grid}{\mathrm{grid}}
\newcommand*{\mini}{\mathrm{small}}
\newcommand*{\mammoth}{\mathrm{big}}
\newcommand*{\hor}{\mathrm{wide}}
\newcommand*{\ver}{\mathrm{tall}}
\newcommand*{\medium}{\mathrm{med}}
\newcommand*{\struct}{\mathrm{struct}}

\DeclareMathOperator{\improveCover}{\mathtt{improve-cover}}

\newcommand*{\lwc}{\mathrm{lwc}}  
\newcommand*{\hwc}{\mathrm{hwc}}  
\newcommand*{\bwc}{\mathrm{bwc}}  
\newcommand*{\wwc}{\mathrm{wwc}}  
\newcommand*{\swc}{\mathrm{swc}}  

\DeclareMathOperator{\remMed}{remove-medium-items}
\DeclareMathOperator{\iterFineParts}{\mathtt{iter-fine-parts}}

\DeclareMathOperator{\constrPartSolve}{\mathtt{constr-part-solve}}

\DeclareMathOperator{\partBig}{\mathtt{part-big}}

\DeclareMathOperator{\greedyStack}{\mathtt{greedy-stack}}

\DeclareMathOperator{\partWide}{\mathtt{part-wide}}

\DeclareMathOperator{\FP}{FP}

\DeclareMathOperator{\ipack}{\mathtt{ipack}}

\DeclareMathOperator{\packSmall}{\mathtt{pack-small}}

\newcommand*{\andtext}{\textup{ and }}

\title{Approximation Algorithms for Generalized Multidimensional Knapsack}
\author{Arindam Khan $^*$ \and Eklavya Sharma $^*$ \and K. V. N. Sreenivas
\thanks{Department of Computer Science and Automation, Indian Institute of Science, Bengaluru, India.
{\tt arindamkhan@iisc.ac.in}, {\tt eklavyas@iisc.ac.in}, {\tt venkatanaga@iisc.ac.in}}}
\date{\empty}
\begin{document}
\maketitle
\setlength{\parindent}{0pt}
\setlength{\parskip}{0.5em}
\begin{abstract}
We study a generalization of the knapsack problem with geometric and vector
constraints. The input is a set of rectangular items, each with an associated profit and $d$ nonnegative weights ($d$-dimensional vector), and  a square knapsack.
The goal is to find a non-overlapping axis-parallel packing of a subset of items into the given knapsack such  that the vector constraints are not violated, i.e., the sum of weights of all the packed items in
any of the $d$ dimensions does not exceed one. We consider two variants of the
problem: $(i)$ the items are not allowed to be rotated, $(ii)$ items can be rotated by 90 degrees.

We give a $(2+\epsilon)$-approximation algorithm for this problem (both versions). In the process, we also study a variant of the maximum generalized assignment problem (Max-GAP), called Vector-Max-GAP, and
design a PTAS for it.
\end{abstract}

\section{Introduction}
The knapsack problem is a fundamental well-studied problem in the field of
combinatorial optimization and approximation algorithms. It is one of
Karp's 21 NP-Complete problems and has been studied extensively for more than a century \cite{kellererBook}.
A common generalization of the knapsack problem is the 2-D geometric knapsack problem. Here, the knapsack is a unit  square
and the items are rectangular objects with designated profits.
The objective is to pack a maximum profit
subset of items such that no two items overlap and all the items are packed
parallel to the axes of the knapsack (called axis-parallel or orthogonal packing).
There are two further variants of the problem depending on whether we are allowed to rotate the
items by 90 degrees or not.
Another well-studied variant of the knapsack problem is $d$-D vector
knapsack problem, where we are given a set of items, each with $d$-dimensional weight vector and a profit, and
a $d$-dimensional bin ($d$ is a constant).
The goal is to select a subset of items of maximum profit such that the sum of all weights
in each dimension is bounded by the bin capacity in that dimension.

In this paper we study a natural variant of knapsack problem, called $(2,d)$ Knapsack, which considers items with both geometric dimensions and the vector dimensions.
In this $2$-D geometric knapsack problem with $d$-D vector constraints
(where $d$ is a constant), or $(2,d)$ KS in short, the input set $I$ consists
of $n$ items. For any item $i\in I$, let $\itemlength{i},\itemheight{i},\itemprofit{i}$
denote the width, height and the profit of the item, respectively, and let $\itemarea{i}=
\itemlength{i}  \itemheight{i}$ denote the area of the item. Also, for
any item $i\in I$, let $\itemweight{j}{i}$ denote the weight of the item in
the $j\Th$ dimension. The objective is to pack a maximum profit subset of items
$J\subseteq I$ into a unit square knapsack in an axis-parallel, non-overlapping manner such
that for any $j\in[d]$, $\sum_{i\in J}\itemweight{j}{i}\le 1$.




We will also study a variant of Maximum Generalized Assignment Problem (Max-GAP).
In Max-GAP problem, we are provided with
a set of machines with designated capacities, and a set of items; an item's
size and value depends on the machine to which it is going to be assigned.
The objective is to assign a subset of items to machines such that the
obtained value is maximized while making sure that
no machine's capacity is breached. We define a variant of the Max-GAP problem
called \vmgap{}. In this problem, we additionally
have a $d$-dimensional weight vector associated with every item and a
$d$-dimensional global weight constraint on the whole setup of machines. The
objective is to find the maximum value obtainable so that no machine's capacity
is breached and the overall weight of items does not cross the global
weight constraint.

\subsection{Our contributions}
Our two main contributions are $(i)$ a $(2+\epsilon)$ approximation algorithm for
the $(2,d)$ KS problem $(ii)$ define a new problem called \vmgap{}
and obtain a PTAS for it.

To obtain the approximation algorithm for $(2,d)$ KS, we first obtain a
structural result using corridor decomposition, which shows that if we consider
an optimal packing and remove either vertical items (height$\gg$width) or horizontal
items (width$\gg$height), we can restructure the remaining packing to possess a
`nice' structure that can be efficiently searched for.

In \cite{l-packing}, which deals with the 2-D geometric knapsack problem,
the authors use the PTAS for the maximum generalized assignment
problem (Max-GAP) to efficiently search for such a nice structure. But due to the presence
of vector constraints, we needed a generalization of the Max-GAP problem which
we call the Vector-Max-GAP problem and obtain a PTAS for it to obtain the desired
approximation algorithm for $(2,d)$ KS.
\subsection{Related work}
The classical knapsack problem admits a fully polynomial time approximation scheme (FPTAS)
\cite{knapsack-ptas,lawler}.
Both 2-D geometric knapsack (2D GK) and 2-D vector knapsack problem are W[1]-hard \cite{GKW19, kulik2010there}, thus do not admit an EPTAS unless W[1]=FPT.  \footnote{EPTAS is a PTAS with running time $f(\epsilon)n^c$ where $c$ is a
constant that does not depend on $\eps$.}
However, $d$-D vector knapsack problem admits a PTAS \cite{vector-knapsack-ptas}.
On the other hand,  it is not known whether the 2D GK is APX-Hard or not, and finding a PTAS for the problem is one of the major open problems in the area.
The current best approximation algorithm \cite{l-packing} for this problem without
rotations achieves an approximation ratio of $17/9+\epsilon$.
If rotations are allowed, the approximation ratio improves to $3/2+\epsilon$.
In a special case known as the \emph{cardinality} case where each
item has a profit of one unit, the current best approximation factor is
$558/325+\epsilon$ without rotations and $4/3+\epsilon$ with rotations (see \cite{l-packing}).
PTASes exist for the \emph{restricted profit density}
case, i.e., the profit/area ratio of each item is upper bounded and lower bounded
by fixed constants (see \cite{bansal2009structural})
and the case where all the items are \emph{small} compared to the dimensions of the bin.
EPTAS  exists for the case where all the items are \emph{squares}
\cite{HeydrichWiese2017}. In the case where we can use \emph{resource augmentation}, i.e.,
we are allowed to increase the width or height of the knapsack
by a factor of $\epsilon$, we can construct a packing which gives the exact optimal
profit in polynomial time (see \cite{HeydrichWiese2017}).
2GK has also been studied under psedopolynomial setting \cite{2dk1, 2dk2}.
The following table summarizes the results for 2GK problem.
\begin{figure}[H]
\centering
\begin{tabular}{|l||c|c|}
\hline
&without rotations&with rotations\\\hline\hline
General case&$17/9+\epsilon$\cite{l-packing}&$3/2+\epsilon$\cite{l-packing}\\\hline
Cardinality case&$558/325+\epsilon$\cite{l-packing}&$4/3+\epsilon$\cite{l-packing}\\\hline
Restricted profit density&PTAS\cite{bansal2009structural}&PTAS\cite{bansal2009structural}\\\hline
Small items&folklore FPTAS (\cref{lem:nfdh-small})&folklore FPTAS (\cref{lem:nfdh-small})\\\hline
Square items&EPTAS \cite{HeydrichWiese2017}&irrelevant\\\hline
Resource Augmentation&exact solution\cite{HeydrichWiese2017}&exact solution\cite{HeydrichWiese2017}\\\hline
\end{tabular}
\caption{The state-of-the-art for the 2-D geometric knapsack problem.}
\end{figure}


A well-studied problem closely related to the knapsack problem is bin packing.
A natural generalization of the bin packing problem is the 2-D geometric bin packing (2GBP)
problem where items are rectangles and bins are unit squares. The current best
asymptotic approximation ratio of $(1+\ln1.5+\eps)$ for 2GBP is due to \cite{bansal2014binpacking}.
Another generalization of bin packing is the vector bin packing problem (VBP) where
each item has multiple constraints. For the current best approximation ratios of VBP, we refer the reader to \cite{BansalE016}.

There exist many other variants of the discussed packing problems such as strip packing
(see \cite{Galvez20, GGIK16}), maximum independent set of rectangles (see \cite{KhanR20}), storage allocation problems \cite{MomkeW20}, etc.
For an extensive survey of related packing problems, see \cite{CKPT17, Khan16}.



\subsection{Organization of the Paper}
In \cref{vec-max-gap}, we study the Vector-Max-GAP problem and
design a PTAS for it.
\cref{sec:mixknap} describes  the  $(2+\epsilon)$ approximation
algorithm for the $(2,d)$ Knapsack problem. In \cref{geomvecks-struct-intro,geomvecks-struct}, we use a structural result involving corridor decomposition and reduce our problem to a \emph{container packing problem}.
In \cref{cont-vecmaxgap,red-cont-vecmaxgap} we will discuss how to model the
container packing problem as an instance of the Vector-Max-GAP problem.
Finally, in \cref{mixknap-algo}, we put everything together to obtain the
$(2+\epsilon)$ approximation algorithm for the $(2,d)$ Knapsack problem.

\section{\vmgap{} problem and a PTAS}
\label{vec-max-gap}

We will formally define the \vmgap{} problem. Let $I$ be a set of $n$ items numbered
1 to $n$ and let $M$ be a set of $k$ machines, where $k$
is a constant. The $j\Th$ machine has a capacity $M_j$. Each item
$i\in I$ has a size of $s_j(i)$, value of $\val_j(i)$ in the $j\Th$ machine
($j\in [k]$). Additionally, each item $i$ also has a weight $w_q(i)$ in the
$q\Th$ dimension ($q\in [d]$, $d$ is a constant). Assume that for all $j\in[k]$,
$q\in[d]$ and $i\in[n]$, $M_j, w_q(i), s_j(i), \val_j(i)$ are all non-negative.

The objective is to assign a subset of items $J\subseteq I$
to the machines such that for any machine $j$,
the size of all the items assigned
to it does not exceed $M_j$. Further, the total weight of the set $J$ in any
dimension $q\in[d]$ must not exceed $W_q$, which is the global weight constraint
of the whole setup in the $q\Th$ dimension. Respecting these constraints,
we would like to maximize the total value of the items in $J$.

Formally, let $J$ be the subset of items picked and $J_j$ be the items
assigned to the $j\Th$ machine ($j\in[k]$).
The assignment is feasible iff the following constraints are satisfied:
\[ \forall q \in [d], \sum_{i \in J} w_q(i) \le W_q \tag{weight constraints} \]
\[ \forall j \in [k], \sum_{i \in J_j} s_j(i) \le M_j \tag{size constraints} \]

Let $\vect{M} = [M_1,M_2,\dots,M_k]$, $\vect{w}(i) = [w_1(i),w_2(i),
\dots,w_d(i)]$, $\vect{s}(i) = [s_1(i),s_2(i),\dots,s_k(i)]$,
$\vect{\val}(i) = [\val_1(i),\val_2(i),\dots,\val_k(i)]$.

Let $\vect{s} = [\vect{s}(1),\vect{s}(2),\dots,\vect{s}(n)]$,
$\vect{w} = [\vect{w}(1),\vect{w}(2),\dots,\vect{w}(n)],
\vect{\val} = [\vect\val(1),\vect\val(2),\dots,\vect\val(n)]$.
$\vect{W} = [W_1,W_2,\dots,W_d]$.

An instance of this problem is given by
$(I, \vect{\val}, \vect{s}, \vect{w}, \vect{M}, \vect{W})$.
We say that the set of items $J$ is feasible for
$(\vect{s}, \vect{w}, \vect{M}, \vect{W})$
iff $J$ can fit in machines of capacity
given by $\vect{M}$ and satisfy the global weight
constraint given by $\vect{W}$
where item sizes and weights are given by $\vect s$ and $\vect w$ respectively.

\subsection{Dynamic-Programming Algorithm for Integral Input}

Consider the case where item sizes and weights are integers.
\WLoG, we can assume that the capacities of machines,
$M_j$ $(j\in [k])$ and weight constraints, $W_q$ $(q\in[d])$, are also integers
(otherwise, we can round them down to the closest integers).

Arbitrarily order the items and number them from $1$ onwards.
Let $\Val(n, \vect{M}, \vect{W})$ be the maximum value obtainable
by assigning a subset of the first $n$ items to $k$ machines with capacities
given by $\vect{M}$ respecting the global weight constraint $\vect{W}$.
We can express $\Val(n, \vect{M}, \vect{W})$ as a recurrence.

\[ \Val(n, \vect{M}, \vect{W}) = \left\{\begin{array}{lr}
-\infty & \textrm{if } \neg(\vect{W} \ge 0 \wedge \vect{M} \ge 0)
\\ 0 & \textrm{if } n = 0
\\ \multicolumn{2}{l}{
{\displaystyle \max\left(\begin{array}{l} \Val(n-1, \vect{M}, \vect{W}),
\\ {\displaystyle \max_{j=1}^k \left( \val_j(n)
+ \Val\left(n-1, \vect{M} - \vect{s}(n)\cdot\vect{e}_j, \vect{W} - \vect{w}(n)\right)\right)} \end{array}\right)}
\; \textrm{else}}
\end{array}\right. \]

$\Val(n, \vect{M}, \vect{W})$ can be computed using dynamic programming.
We can find the subset of items that gives this much value and it is also easy to ensure that no item assigned to a machine has value 0 in that machine. There are $n\prod_{j=1}^k(M_j+1)\prod_{q=1}^d(W_q+1)$ items in the state space
and each iteration takes $\Theta(d + k)$ time.
Therefore, time taken by the dynamic programming solution
is $\Theta\left(n(d+k)\prod_{j=1}^k (M_j+1)\prod_{q=1}^d (W_q+1)\right)$.

\subsection{Optimal Solution with Resource Augmentation}

Let $\vect{\mu} = [\mu_1, \mu_2, \ldots, \mu_k]$ and $\vect{\delta} = [\delta_1, \delta_2, \ldots, \delta_d]$
be vectors whose values will be decided later. For $j\in [k]$, define $s_j'(i) = \ceil{s_j(i)/\mu_j},
M_j' = \floor{M_j/\mu_j} + n$.
For $q\in[d]$, define $w_q'(i) = \ceil{w_q(i)/\delta_q},W_q' = \floor{W_q/\delta_q} + n$.
\begin{lemma}
\label{thm:ra1}
Let $J$ be feasible for $(\vect{s}, \vect{w}, \vect{M}, \vect{W})$.
Then $J$ is also feasible for $(\vect{s'}, \vect{w'}, \vect{M'}, \vect{W'})$.
\end{lemma}
\begin{proof}
For any dimension $q\in[d]$,
$\sum_{i \in J} w_q'(i)
= \sum_{i \in J} \ceil{w_q(i)/\delta_q}
\le \sum_{i \in J} \left(\floor{w_q(i)/\delta_q}+1\right)$.
\[\sum_{i \in J} \left(\floor{w_q(i)/\delta_q}+1\right)
\le |J| + \floor{ (1/\delta_q) \sum_{i \in J} w_q(i) }
\le n + \floor{W_q/\delta_q} = W_q'\]

Let $J_j$ be the items in $J$ assigned to the $j\Th$ machine. Then
$\sum_{i \in J_j} s_j'(i)
= \sum_{i \in J_j} \ceil{s_j(i)/\mu_j}$.
\[
\sum_{i \in J_j} \left(\floor{s_j(i)/\mu_j}+1\right)
\le |J_j| + \floor{ (1/\mu_j) \sum_{i \in J_j} s_j(i) }
\le n + \floor{M_j/\mu_j} = M_j'
\qedhere\]
\end{proof}

\begin{lemma}
\label{thm:ra2}
Let $J$ be feasible for $(\vect{s'}, \vect{w'}, \vect{M'}, \vect{W'})$.
Then $J$ is also feasible for
$(\vect{s}, \vect{w}, \vect{M} + n\vect{\mu}, \vect{W} + n\vect{\delta})$.
\end{lemma}
\begin{proof} For all $q\in[d]$
\[
\sum_{i \in J} w_q(i)
\le \sum_{i \in J} \delta_q w_q'(i)
\le \delta_q W_q'
= \delta_q\left(\floor{W_q/\delta_q} + n\right)
\le W_q + n\delta_q
\]
Let $J_j$ be the items in $J$ assigned to the $j\Th$ machine.
\[
\sum_{i \in J_j} s_j(i)
\le \sum_{i \in J_j} \mu_j s_j'(i)
\le \mu_j M_j'
= \mu_j\left(\floor{M_j/\mu_j} + n\right)
\le M_j + n\mu_j
\qedhere\]
\end{proof}

Let $\mu_j = \epsilon M_j/n$ and $\delta_q = \epsilon W_q/n$ for all $q\in[d]$ and $j\in[k]$.
Let $J^*$ be the optimal solution to $(I, \vect{\val}, \vect{s}, \vect{w}, \vect{M}, \vect{W})$.
Let $\widehat{J}$ be the optimal solution to $(I, \vect{\val}, \vect{s'}, \vect{w'}, \vect{M'}, \vect{W'})$.
By \Cref{thm:ra1}, $\val(\widehat{J}) \ge \val(J^*)$.
By \Cref{thm:ra2}, $\widehat{J}$ is feasible for
$(\vect{s}, \vect{w}, (1+\epsilon)\vect{M}, (1+\epsilon)\vect{W})$. Also, observe
that $|M_j'| \le n + M_j/\mu_j = n(1 + 1/\epsilon)$ is a
polynomial in $n$. Similarly, $|W_q'| \le n(1 + 1/\epsilon)$ and hence
the optimal solution to $(I, \vect{\val}, \vect{s'}, \vect{w'}, \vect{M'}, \vect{W'})$
can be obtained using the dynamic-programming algorithm
in polynomial time.
Therefore, the optimal solution to $(I, \vect{\val}, \vect{s'}, \vect{w'}, \vect{M'}, \vect{W'})$
can be obtained using the dynamic-programming algorithm
in time ${\displaystyle \Theta\left((d+k)n^{d+k+1}/\epsilon^{d+k}\right)}$.

Let us define a subroutine $\operatorname{assign-res-aug}_{\epsilon}
(I, \vect{\val}, \vect{s}, \vect{w}, \vect{M}, \vect{W})$ which takes as input set
$I$ with associated values, $\vect{\val}$, and gives as output the optimal feasible
solution to $(\vect{s}, \vect{w}, (1+\epsilon)\vect{M}, (1+\epsilon)\vect{W})$.

\subsection{Trimming}

Consider a set $I$ of items where each item has length $s(i)$ and profit $p(i)$
(in this subsection, $I$ is an instance of the knapsack problem
instead of the \vmgap{} problem).

Suppose for all $i \in I, s(i) \in (0, \epsilon]$ and $s(I) \le 1+\delta$.
We'll show that there exists an $ R \subseteq I$ such that
$s(I-R) \le 1$ and $p(R) < (\delta + \epsilon)p(I)$.
We call this technique of removing a low-profit subset from $I$
so that it fits in a bin of length 1 trimming.

Arbitrarily order the items and arrange them linearly in a bin of size $1+\delta$.
Let $k = \floor{1/(\delta + \epsilon)}$.
Create $k+1$ intervals of length $\delta$ and $k$ intervals of length $\epsilon$.
Place $\delta$-intervals and $\epsilon$-intervals alternately.
They will fit in the bin because $(k+1)\delta + k\epsilon = \delta + (\delta + \epsilon)\floor{1/(\delta + \epsilon)} \le 1 + \delta $

Number the $\delta$-intervals from $0$ to $k$ and let $S_i$ be the set of items
intersecting the $i\Th$ $\delta$-interval. Note that, all $S_i$ are mutually disjoint.
Let $i^* = \argmin_{i=0}^k p(S_i)$.
\[ p(S_{i^*}) = \min_{i=0}^k p(S_i) \le \frac{1}{k+1} \sum_{i=0}^k p(S_i)
\le \frac{p(I)}{\floor{1/(\delta + \epsilon)} + 1}
< (\delta + \epsilon)p(I) \]

Removing $S_{i^*}$ will create an empty interval of length $\delta$
in the bin, and the remaining items can be shifted so that they fit in a bin of length 1.

\subsection{Packing Small Items}

Consider a \vmgap{} instance $(I, \vect{\val}, \vect{s}, \vect{w}, \vect{M}, \vect{W})$.
Item $i$ is said to be $\epsilon$-small for this instance iff
$\vect{w}(i) \le \epsilon \vect{W}$ and
for all $ j \in [k], (s_j(i) \le \epsilon M_j \textit{ or } \val_j(i) = 0)$.
A set $I$ of items is said to be $\epsilon$-small iff each item in $I$ is $\epsilon$-small.

Suppose $I$ is $\epsilon$-small.
Let $J \subseteq I$ be a feasible solution to
$(I, \vect{\val}, \vect{s}, \vect{w}, (1+\epsilon)\vect{M}, (1+\epsilon)\vect{W})$.
Let $J_j$ be the items assigned to the $j\Th$ machine.

For each $j \in [k]$, use trimming on $J_j$ for sizes $s_j$
and then for each $j \in [d]$, use trimming on $J$ for weights $w_j$.
In both cases, use $\epsilon:=\epsilon$ and $\delta:=\epsilon$.
Let $R$ be the removed items and $J' = J-R$ be the remaining items.
Total value lost is less than $2\epsilon(d+1)\val(J)$
and $J'$ is feasible for $(\vect{s}, \vect{w}, \vect{M}, \vect{W})$.

Therefore, any resource-augmented solution $J$ of small items can be transformed
to get a feasible solution $J'$ of value at most $(1-2(d+1)\epsilon)\val(J)$.

\subsection{A Structural Result}

\begin{theorem}
\label{thm:struct}
Let $J$ be a feasible solution to $(I, \vect{\val}, \vect{s}, \vect{w}, \vect{M}, \vect{W})$.
Let $J_j \subseteq J$ be the items assigned to the $j\Th$ machine.
Then for all $ \epsilon > 0$, there exist sets $X$ and $Y$ such that
$|X| \le (d+k)/(\epsilon^2)$ and $\val(Y) \le \epsilon\cdot \val(J)$ and
\begin{align*}
\forall j \in [k], \forall i \in J_j - X - Y, &\; s_j(i) \le \epsilon \left(M_j - s_j\left(X \cap J_j\right)\right)
\\ \forall i \in J - X - Y, &\; \vect{w}(i) \le \epsilon \left(\vect{W} - \vect{w}(X)\right)
\end{align*}
\end{theorem}

\begin{proof}
Let $P_{1,j} = \{i \in J_j: s_j(i) > \epsilon M_j \}$, where $j \in [k]$,
and $Q_{1,q} = \{i \in J: w_q(i) > \epsilon W_q \}$, where $q \in [d]$.
We know that $s_j(P_{1,j}) > \epsilon M_j |P_{1,j}|$. Therefore, $|P_{1,j}| \le \frac{1}{\epsilon}$.
Also, $w_q(Q_{1,q}) > \epsilon W_q |Q_{1,q}|$. Therefore, $|W_{1,q}| \le \frac{1}{\epsilon}$.
Let $R_1 = \left( \bigcup_{j=1}^k P_{1,j} \right) \cup \left( \bigcup_{q=1}^d Q_{1,q} \right)$.
$R_1$ is therefore the set of items in $J$ that are in some sense `big'. Note that
$|R_1| \le (d+k)/\epsilon$.

If $\val(R_1) \le \epsilon \cdot\val(J)$, set $Y = R_1$ and $X = \{\}$ and we're done.
Otherwise, set
$P_{2,j} = \{i \in J_j - R_1: s_j(i) > \epsilon (M_j - s_j(R_1 \cap J_j)) \}$,
$Q_{2,q} = \{i \in J - R_1: w_q(i) > \epsilon (W_q - w_q(R_1)) \}$,
and $R_2 = \left( \bigcup_{j=1}^k P_{2,j} \right) \cup \left( \bigcup_{q=1}^d Q_{2,q} \right)$.
If $\val(R_2) \le \epsilon \cdot\val(J)$, set $Y = R_2$ and $X = R_1$ and we're done.
Otherwise, set
$P_{3,j} = \{i \in J_j - R_1 - R_2: s_j(i) > \epsilon (M_j - s_j((R_1 \cup R_2) \cap J_j)) \}$,
$Q_{3,q} = \{i \in J - R_1 - R_2: w_q(i) > \epsilon (W_q - w_q(R_1) - w_q(R_2)) \}$,
and $R_3 = \left( \bigcup_{j=1}^k P_{3,j} \right) \cup \left( \bigcup_{q=1}^d Q_{3,q} \right)$.
If $\val(R_3) \le \epsilon \cdot\val(J)$, set $Y = R_3$ and $X = R_1 \cup R_2$ and we're done.
Otherwise, similarly compute $R_4$ and check if $\val(R_4) \le \epsilon \cdot\val(J)$, and so on.
Extending the similar arguments about  $\abs{R_1}$, it follows that for all $t>0$, $|R_t| \le (d+k)/\epsilon$.

Since every $R_t$ $(t > 0)$ is disjoint, there will be some $R_T$ such that $\val(R_T) \le \epsilon\cdot\val(J)$.

Now set $Y = R_T$ and $X = R_1\cup\ldots \cup R_{T-1}$.
We can see that $|X| \le \sum_{t=1}^T |R_T| \le T(d+k)/\epsilon \le (d+k)/\epsilon^2$ and
$\val(Y) = \val(R_T) \le \epsilon\cdot \val(J)$. Note that all items in $J - X - Y$ are
small because of the way $R_T$ was constructed. Hence it follows that
\[ \forall j \in [k], \forall i \in J_j - X-Y, s_j(i) \le \epsilon (M_j - s_j(X \cap J_j)) \]
\[ \forall i \in J - X-Y, \vect{w}(i) \le \epsilon \left(\vect{W} - \vect{w}(X)\right) \qedhere\]
\end{proof}

\subsection{PTAS for \vmgap{}}

Let $J^*$ be an optimal assignment for $(I, \vect{\val}, \vect{s}, \vect{w}, \vect{M}, \vect{W})$.
Let $J_j^* \subseteq J^*$ be the items assigned to the $j\Th$ machine.

By \Cref{thm:struct}, $J^*$ can be partitioned into sets $X^*$, $Y^*$ and $Z^*$ such that
$|X^*| \le \frac{d+k}{\epsilon^2}$ and $\val(Y^*) \le \epsilon\cdot \val(J^*)$.
Let $\vect{W}^* = \vect{W} - \vect{w}(X^*)$ and $M_j^* = M_j - s_j(X^* \cap J_j^*)$.
Then $Z^*$ is $\epsilon$-small for $(\vect{\val}, \vect{s}, \vect{w}, \vect{M}^*, \vect{W}^*)$.

For a set $S$, define $\Pi_k(S)$ as the set of $k$-partitions of $S$. Then
the pseudo code in \cref{algo:cvks-ptas} provides a PTAS for the \vmgap{}
problem.

\begin{algorithm}[H]
\caption{$\vmg(I, \vect{p}, \vect{s}, \vect{w}, \vect{M}, \vect{W})$:
PTAS for Vector Max GAP}
\begin{algorithmic}
\State $J_{\textrm{best}} = \{\}$.
\For{$X \subseteq I$ such that $|X| \le (d+k)/\epsilon^2$}
    \For{$(X_1, X_2, \ldots, X_k) \in \Pi_k(X)$}
        \State $\vect{W}' = \vect{W} - \vect{w}(X)$
        \State $M_j' = M_j - s_j(X_j)$ for each $j \in [k]$.
        \State $\val_j'(i) = \begin{cases} \val_j(i) & \textrm{if } s_j(i) \le \epsilon M_j'
            \\ 0 & \textrm{otherwise} \end{cases}$ for each $i \in I-X$.
        \State \Comment{$I-X$ is $\epsilon$-small for $(\vect{\val'}, \vect{s}, \vect{w}, \vect{M}', \vect{W}')$}
        \State $Z' = \operatorname{assign-res-aug}_{\epsilon}
            (I-X, \vect{\val'}, \vect{s}, \vect{w}, \vect{M}', \vect{W}')$.
        \State Trim $Z'$ to get $Z$ so that $Z$ is feasible for $(\vect{s}, \vect{w}, \vect{M'}, \vect{W'})$.
        \State $J = X \cup Z$
        \If{$\val(J) > \val(J_{\textrm{best}})$}
            \State $J_{\textrm{best}} = J$
        \EndIf
    \EndFor
\EndFor
\State \Return $J_{\textrm{best}}$
\end{algorithmic}
\label{algo:cvks-ptas}
\end{algorithm}

\textbf{Correctness}:
Since $Z$ is feasible for $(\vect{s}, \vect{w}, \vect{M}', \vect{W}')$,
$X \cup Z$ is feasible for $(\vect{s}, \vect{w}, \vect{M}, \vect{W})$.

\textbf{Approximation guarantee}:\\
For some iteration of \Cref{algo:cvks-ptas},
$X = X^*$ and $X_j = X^* \cap J_j^*$.
When that happens, $\vect{W}' = \vect{W}^*$ and $\vect{M}' = \vect{M}^*$.
Let $\val^*$ be the maximum value $\epsilon$-small assignment of items to the
machines with capacities given by $\vect{M}'$ and over all weight constraints $\vect{W}'$.
Therefore, $\val^* \ge \val(Z^*)$.

To try to find an optimal assignment of small items,
we'll forbid non-small items to be assigned to a machine.
To do this, if for item $i$, $s_j(i) > \epsilon M_j'$, set $\val_j(i)$ to 0.
Using our resource-augmented \vmgap{} algorithm,
we get $\val(Z') \ge \val^*$.
By the property of trimming, $\val(Z) \ge (1-2(d+1)\epsilon)\val(Z')$.
\[ \val(J_{\textrm{best}})
\ge \val(X^*) + \val(Z)
\ge \val(X^*) + (1-2(d+1)\epsilon)\val(Z^*)
\ge (1-2(d+1)\epsilon)(1-\epsilon)\val(J^*)
\]
This gives us a $(1-(2d+3)\epsilon)$-approx solution. The running time can be
easily seen to be polynomial as $\operatorname{assign-res-aug}$ runs in
polynomial time and the number of iterations of the outer loop in \Cref{algo:cvks-ptas}
is polynomial in $n$ and for one iteration of the outer loop, the inner loop runs
at most constant number of times.

\section{Algorithm for \geomvec{2}{\texorpdfstring{$d$}{d}} Knapsack Problem}
\label{sec:mixknap}

In this section, we will obtain a $(2+\epsilon)$-approximation algorithm for the
\geomvec{2}{$d$} knapsack problem both with and without rotations. First, we
show the algorithm for the case without rotations. An algorithm for the case
rotations are allowed is similar except for some small changes.

Let $I$ be a set of $n$ \geomvecdim{2}{$d$} items.
We are given a \geomvecdim{2}{$d$} knapsack
and we would like to pack a high profit subset of $I$ in the knapsack.
Let us denote this optimal profit by $\optgvks(I)$. Let
$\lengthinp=[\itemlength{1},\dots,\itemlength{n}]$,
$\heightinp=[\itemheight{1},\dots,\itemheight{n}]$,
$\profitinp=[\itemprofit{1},\dots,\itemprofit{n}]$. For an item $i\in I$, let
$\weightvec{i} = [\itemweight{1}{i}, \dots, \itemweight{d}{i}]$ and
let $\weightinp = [\weightvec{1}, \dots, \weightvec{n}]$.

In the whole section, a \emph{container} is a rectangular region inside the
knapsack. For our purposes, every container can be one of the four types:
\emph{large, wide, tall, area}. A large container can contain at most one item. An area
container can only contain items that are $\epsilon$-small for the container i.e.
an item can be packed into an area container of width $w$ and height $h$
only if the item has width at most $\epsilon w$ and height at most $\epsilon h$.
In a wide (resp. tall) container, items must be stacked up one on top
of another (resp. arranged side by side). We also require that the containers do
not overlap amongst themselves and no item partially overlaps with any of
the containers.

We also use the notation $\bigOhConst{\epsilon}$ in place of $\bigOh{f(\epsilon)}$,
where $\epsilon$ is an arbitrary constant and $f(\cdot)$ is a function which solely
depends on the value of $\epsilon$, when we do not explicitly define what $f(\cdot)$ is.
Similarly, $\bigOhConst{\epsilon_1, \epsilon_2}$ represents
$\bigOh{f(\epsilon_1,\epsilon_2)}$ and so on.

\subsection{A Structural Result}
\label{geomvecks-struct-intro}
Consider a set of items $S$ that are packed in a knapsack.
We now state a structural result, inspired by \cite{l-packing},
where only a subset of items $S'\subseteq S$
is packed into the knapsack. We may lose some profit but the packing has a nice
structure which can be searched for, efficiently.

\begin{theorem}
\label{structresult}
Let $S$ denote a set of items that can be feasibly packed into a knapsack
and let $0<\epsilon<1$ be any small constant.
Then there exists a subset $S' \subseteq S$ such that
$\itemprofit{S'} \ge (1/2-\epsilon)\cdot \itemprofit{S}$.
The items in $S'$ are packed into the knapsack in containers.
Further, the number of containers formed is $\bigOhConst{\epsilon}$
and their widths and heights belong
to a set whose cardinality is $\poly(|S|)$ and moreover, this set can
be computed in time $\poly(|S|)$.
\end{theorem}

Now, let us go back to our original problem instance $I$. Let $I_{\OPT}$
be the set of items packed into the knapsack in an optimal packing $\mathcal{P}$.

Let us apply \Cref{structresult} to the set of items $I_{\OPT}$ with
$\epsilon := \epsilon_{\struct}$ ($\epsilon_{\struct}$ will be defined later).
Let $I_{\OPT}'$ be the resulting analog of $S'$ in the theorem
(there can be many candidates for $I_{\OPT}'$ but let us pick one). Therefore,
\begin{align}
\itemprofit{I_{\OPT}'}\ge (\frac{1}{2}-\epsilon_{\struct})\cdot \itemprofit{I_{\OPT}} = (\frac{1}{2}-\epsilon_{\struct})\cdot \optgvks{(I)} \label{halfprofit}
\end{align}

\subsection{Proof of the Structural Result}
\label{geomvecks-struct}

In this subsection, we will prove \Cref{structresult}.

The strategy to prove the theorem is to use the corridor decomposition scheme,
essentially taken from \cite{adamaszek2015knapsack}. First, we assume
that we can remove $\bigOhConst{1/\epsilon}$ number of items at the cost
of zero profit loss from the originally packed items. Under this assumption, we show
that we lose at most half profit by our restructuring. Finally, we show how to get
rid of this assumption by using shifting argumentations.

\textbf{Removing Medium items}: Let $\epssmall, \epslarge$ be two fixed constants such that $\epslarge > \epssmall$.
We partition the items in $S$ based on the values of $\epssmall$ and $\epslarge$ as
follows:
\begin{itemize}
    \item $S_S = \left\{i\in S\::\:\itemlength{i}\le \epssmall \andtext \itemheight{i}\le \epssmall\right\}$
    \item $S_B = \left\{i\in S\::\:\itemlength{i}> \epslarge \andtext \itemheight{i}> \epslarge\right\}$
    \item $S_W = \left\{i\in S\::\:\itemlength{i}> \epslarge \andtext \itemheight{i}\le \epssmall\right\}$
    \item $S_T = \left\{i\in S\::\:\itemlength{i}\le \epssmall \andtext \itemheight{i}> \epslarge\right\}$
    \item $S_{\medium} = \left\{i\in S\::\:\itemlength{i}> \epssmall \andtext \itemlength{i}\le \epslarge
                                 \quad \textup{OR}\quad \itemheight{i}\le \epslarge \andtext \itemheight{i}> \epssmall\right\}$
\end{itemize}

We call the items in $S_S$ as \emph{small} items as they are small in both the
dimensions. Similarly, we call the items in $S_B, S_W, S_T, S_{\medium}$ as
\emph{big, wide, tall, medium} respectively.

By standard arguments, it is possible to choose the constants $\epssmall$ and
$\epslarge$ such that the total profit of all the medium items is at most
$\epsilon\cdot \itemprofit{S}$. Hence, we can discard the items in $S_{\medium}$ from $S$
while losing a very small profit. We omit the proof as it is present in many articles
on packing (see, for example, \cite{l-packing}).

\subsubsection{Corridors}

First, let us define what a subcorridor is. A subcorridor is just a rectangle in
the 2D coordinate system with one side longer than $\epslarge$ and the other side having a length
of at most $\epslarge$. A subcorridor is called wide (resp. tall) if the longer side is
parallel to the \xaxis{} (resp. \yaxis{}).

A corridor is just a subcorridor or a union of at most $1/\epsilon$
subcorridors such that each wide (resp. tall)
subcorridor overlaps with exactly two tall (resp. wide) subcorridors, except for
at most two subcorridors which are called the \emph{ends} of the corridors and can overlap
with exactly one subcorridor. The overlap between any two subcorridors should be
in such a way that one of their corners coincide.

\subsubsection{Corridor Decomposition}

For now, let's consider a generic packing of a set of items $S$ that can
contain big, small, wide, tall items.

Since the number of big items packed can be at most a constant, and since we
assumed that we can remove a constant number of items at no profit loss, let us
remove all the big items from the packing. Let's also get rid of the
small items for now (we will pack the small items later). Hence, we are left with
wide and tall items packed in the knapsack. Let's name these items as
\emph{skewed} items and denote the set of these items by $S_\skewed$.

\begin{lemma}[Corridor Packing lemma \cite{adamaszek2015knapsack}]
\label{corr-packing}
There exist non-overlapping corridor regions in the knapsack such that
we can partition $S_\skewed$ into sets $S_\corr$,
$S_\cross^\nice$, $S_\cross^\bad$ such that
\begin{itemize}
    \item $\abs{S_\cross^\nice} = \bigOhConst{1/\epsilon}$
    \item $\itemprofit{S_\cross^\bad} \le \bigOhConst{\epsilon}\cdot\itemprofit{S_\skewed}$
    \item Every item in $S_\corr$ is fully contained in one of the corridors.
    The number of corridors is $\bigOhConst{1/\epsilon,1/\epsilon_{\mammoth}}$ and in each
    corridor, the number of subcorridors is at most $1/\epsilon$.
    \item Each subcorridor has length at least $\epsilon_{\mammoth}$ and
    breadth less than $\epsilon_{\mammoth}$ (assuming length to denote the longest side
    and breadth to denote the smallest side).
\end{itemize}
\end{lemma}
\begin{remark}
\label{subcorr-total}
The total number of subcorridors after the corridor partition is
$\bigOhConst{1/\epsilon,1/\epsilon_{\mammoth}}$.
\end{remark}
We also remove items in $S_\cross^\nice$ since they are at most a constant in
number and items in $S_\cross^\bad$ since their total profit is very less.

The last point of the lemma ensures that a skewed item is completely contained
in at most one subcorridor. Hence for every skewed item contained in a corridor,
we can tell which subcorridor it \emph{belongs} to. The last point also ensures that,
there can not be a subcorridor which completely contains both wide and tall
items. This fact allows us to label each subcorridor as wide or tall: If a
subcorridor contains only wide (resp. tall) items, we say that it is a
wide (resp. tall) subcorridor.

\textbf{Removing either wide or tall items}: Now, we will simplify the above
structure of skewed items while losing at most half of the original profit.

Assume \wLoG{} that the total profit of wide items is at least as much as the total
profit of the tall items. Hence, if we get rid of the tall items, we lose
at most half of the original profit. With this step, we can also remove all the
tall subcorridors since they are all empty. We are left with wide items
packed in wide subcorridors. Since the subcorridors are just rectangles and
they no longer overlap, we just call these subcorridors as \emph{boxes}.

Next, we will describe how to reorganize the items in these boxes into containers
at a very marginal profit loss.

\subsubsection{Reorganizing Boxes into Containers}

\label{box-to-cont}
Note that the boxes contain only wide items. And since the width of
the wide items is lower bounded by $\epsilon_{\mammoth}$, we can create an empty
wide strip in each box at the loss of very small profit and at most a constant
number of items. This is described in the following lemma.

\begin{lemma}
\label{strip-removal}
Let $B$ be a box of dimensions $a\times b$ such that each item contained in it
has a width of at least $\mu$, where $\mu$ is a constant. Let the total profit
packed inside the item be $P$ and let $\epsilon_\boxrect$ be a
small constant. Then, it is possible to pack all the items barring a set of
items of profit at most $\epsilon_\boxrect\cdot P$ and a set of
$\bigOhConst{1/\mu}$ number of items into a box of dimensions $a\times(1-\epsilon_\boxrect)b$.
\end{lemma}
\begin{proof}
Assume the height of the box is along the \yaxis{} and width of the box is along
the \xaxis{} and say, the bottom left corner of the box is situated at the origin. Draw lines
at $y=y_i$ for $i\in \left\{1,\dots,\floor{1/\epsilon_\boxrect}\right\}$
such that $y_i=i\cdot\epsilon_\boxrect b$.
These lines will partition the box into $\ceil{1/\epsilon_\boxrect}$ regions. For all
$i\in\left\{1,\dots,\ceil{1/\epsilon_\boxrect}\right\}$, let $s_i$ be
the set of items completely contained in the $i\Th$ region. Then, there
must exist some region $j$ such that $\itemprofit{s_j}\le (1/\ceil{1/\epsilon_\boxrect})\cdot P$.
We will remove all the items in $s_j$. Also, the number of items
partially overlapping with the region $j$ can be at most $2/\mu$, which is
$\bigOhConst{1/\mu}$. By removing these items,
we create an empty strip of size $a\times \epsilon_\boxrect b$
and hence the remaining items can be
packed in a box of size $a\times(1-\epsilon_\boxrect)b$.
\end{proof}

We apply the above lemma to each and every box with small enough
$\epsilon_\boxrect$ and
$\mu=\epsilon_{\mammoth}$. Consider a box $B$ and let
$S_{\boxrect}$ be the items packed inside $B$.
Also, let $a\times b$ be the dimensions of the box.
Let $S_{\boxrect}'$ be the set of items
in the knapsack after performing the steps in \Cref{strip-removal}. $S_{\boxrect}'$ can
be packed within the box $a\times (1-\epsilon_\boxrect)b$ and we are left with some
empty space in the vertical direction. Hence, we can apply the resource
augmentation lemma, which is taken from \cite{l-packing}.

\begin{lemma}[Resource Augmentation Lemma]
\label{ra-lemma}
Let $I$ be a set of items packed inside a box of dimensions $a\times b$ and let
$\epsilon_\ra$ be a given constant. Then there exists a container packing of
a set $I'\subseteq I$ inside a box of size $a\times (1+\epsilon_\ra)b$ such that
\begin{itemize}
    \item $\itemprofit{I'}\ge(1-\bigOhConst{\eps_\ra})\cdot \itemprofit{I}$
    \item The number of containers is $\bigOhConst{1/\epsilon_\ra}$ and the
    dimensions of the containers belong to a set of cardinality
    $\mathcal{O}\left(\abs{I}^{\bigOhConst{1/\epsilon_\ra}}\right)$.
    \item The total area of the containers is at most $\areasym(I)+\epsilon_\ra\cdot ab$
\end{itemize}
\end{lemma}

Applying the above lemma to the box $B$, we obtain a packing where $S_{\boxrect}'$ is
packed into box of size $a\times (1-\epsilon_\boxrect)(1+\epsilon_\ra)b$.
We will choose $\epsilon_\boxrect,\epsilon_\ra$ such that the product
$(1-\epsilon_\boxrect)(1+\epsilon_\ra) < (1-2\epsilon)$.

\begin{lemma}
\label{cont-area-bound}
The total area of containers obtained is at most
\[
    \min\left\{(1-2\epsilon),\areasym(S_\corr)+\epsilon_\ra\right\}
\]
\end{lemma}
\begin{proof}
The second upper bound directly follows from the last point of \Cref{ra-lemma}:
Area of the containers in a box $B$ of dimensions $a\times b$ is at most
$\areasym(S_{\boxrect})+\epsilon_\ra\cdot ab$. Summing over all boxes we get the bound
$\areasym(S_\corr)+\epsilon_\ra$.

For the first upper bound, observe that every box of dimensions $a\times b$ is
shrunk into a box of dimensions $a\times (1-2\epsilon)b$. Since all the
original boxes were packable into the knapsack, the total area of the boxes
(and hence the containers) after shrinking is at most $(1-2\epsilon)$ times the
area of the knapsack itself, which is 1.
\end{proof}

To this end, we have completed packing all the items into containers
(except the constant number of discarded items and the small items). We will
analyze the number of containers created. A box was partitioned into
$\bigOhConst{1/\epsilon_\ra}$ containers as given by \cref{ra-lemma} and
boxes are nothing but subcorridors whose number is
$\bigOhConst{1/\epsilon,1/\epsilon_{\mammoth}}$ by \cref{subcorr-total}. Hence
the total number of containers formed
is $\bigOhConst{1/\epsilon,1/\epsilon_{\mammoth},1/\eps_\ra}$. This fact is
going to be crucial in order to pack the small items.

\subsubsection{Packing Small Items}
\label{pack-small}

In this subsection, we show how to pack small items that have been removed from the
packing temporarily. Let the set of small items be denoted by $S_\mini$.
Let $\epsilon_{\grid}=1/\ceil{\epsilon/\epsilon_{\mini}}$.

Let us define a uniform grid $\mathcal{G}$ in the knapsack where grid lines are equally
spaced at a distance of $\epsilon_\grid$. It is easy to see that the number of
grid cells formed is at most $1/\epsilon_\grid^2$ which is
at most a constant. We mark a grid cell as \emph{free} if it has no overlapping
with any container and \emph{non-free} otherwise. We delete all the non-free cells
and the free cells will serve as the area containers that we use to pack the
small items.

\begin{lemma}
For some choice of $\epsilon_{\mini}$ and $\epsilon_\ra$, the total area of non-free cells is
at most
\[
    \min\left\{(1-\epsilon),\areasym(S_\corr)+3\epsilon^2\right\}
\]
\end{lemma}
\begin{proof}
Let $A$ be the total area of containers and $k$ be the number of containers. The
total area of cells completely covered by containers is at most $A$. The partially
overlapped cells are the cells that intersect the edges of the containers. Since,
the number of containers is $k$ and each container has 4 edges and each edge can
overlap with at most $1/\epsilon_\grid$ number of cells,
the area of completely and partially overlapped cells is at most
\[
    A+4k\cdot\frac{1}{\epsilon_\grid}\cdot\epsilon_{\grid}^2
        = A+4k\cdot \frac{1}{\ceil{\epsilon/\epsilon_{\mini}}}
\]
As we noted before, $k=\bigOhConst{1/\epsilon,1/\epsilon_{\mammoth},1/\eps_\ra}$
which doesn't depend on $\eps_\mini$ in some sense.
Hence, we can choose a very small $\eps_\mini$ (when compared to $\epsilon,
\epsilon_{\mammoth},\eps_\ra$), and ensure that the above quantity
is at most $A+2\epsilon^2$.

By \Cref{cont-area-bound}, the value of $A$ is bounded by
$
    \min\left\{(1-2\epsilon),\areasym(S_\corr)+\epsilon_\ra\right\}.
$

Hence, the total area of deleted cells is at most
$
    \min\left\{(1-2\epsilon+2\epsilon^2),\areasym(S_\corr)+\epsilon_\ra+2\epsilon^2\right\}.
$

For $\epsilon<1/2$ and by choosing $\epsilon_{\ra}<\epsilon^2$, we get
the desired result.
\end{proof}

We now show that there is a significant area of free cells left to pack the small
items profitably.

\begin{lemma}
The area of free cells is at least $(1-3\epsilon)\areasym(S_{\mini})$
\end{lemma}
\begin{proof}
Since $S_{\mini}$ and $S_\corr$ were packable in the knapsack,
$\areasym(S_{\mini}\cup S_{\corr}) \le 1$.

Now if $\areasym(S_\mini)\ge \epsilon$, then the area of free cells is at least
$1-\areasym(S_{\corr})-3\epsilon^2 \ge \areasym(S_{\mini})-3\epsilon^2 \ge (1-3\epsilon) \areasym(S_{\mini})$

On the other hand, if $\areasym(S_{\mini}) < \epsilon$, then the area of free cells is
at least $\epsilon$ which in turn is at least $\areasym(S_{\mini})$.
\end{proof}

Each free cell has dimensions $\epsilon_\grid\times \epsilon_\grid$ and a
small item has each side at most $\epsilon_\mini$. Hence, all small items
are $\epsilon$-small for the created free cells and these can be used to pack
the small items. Using NFDH, we can pack a very high profit in the area
containers as described in the following lemma.

\begin{lemma}
Let $I$ be a set of items and let there be $k$ identical area containers
such that each item $i\in I$ is $\epsilon$-small for every container
and the whole set $I$ can be packed in these containers. Then there
exists an algorithm which packs a subset $I'\subseteq I$ in these area containers
such that $\itemprofit{I'} \ge (1-2\epsilon)\itemprofit{I}$.
\end{lemma}
\begin{proof}
\WLoG, assume that each container has dimensions $1\times 1$. So, we can
assume that every item
$i\in I$ has a width of at most $\epsilon$ and a height of at most $\epsilon$.
Order the items in $I$ in non-increasing ratio of profit/area (we call this
profit density) breaking ties
arbitrarily and also order the containers arbitrarily.

Start packing items into the containers using NFDH. If we are able to pack all the
items in $I$, we are done. Otherwise, consider an arbitrary container $C$. Let $I_C$ be
the items we packed in $C$ and let $i_C$ be the item we could not pack in $C$.
Then by \Cref{lem:nfdh-small}, $\areasym(I_C\cup \{i_C\}) > (1-\epsilon)^2\areasym(C) = (1-\epsilon)^2$.
But since $\areasym(i_C) \le \epsilon^2$, $\areasym(I_C)>1-2\epsilon$. Hence, we have
packed at least $(1-2\epsilon)$ fraction of the total area of the containers with
the densest items and thus the claim follows.
\end{proof}

To this end, we have packed at least $(1/2-\bigOhConst{\epsilon})$ fraction of the
total profit in the original packing assuming that we can leave out a constant
number of items at no cost.

\subsubsection{Shifting Argumentation}

We assumed that we can drop $\bigOhConst{1/\epsilon}$ number of items from $S$ at
no profit loss. But this may not be true because they
can carry significant amount of profit
in the original packing. The left out constant number of items are precisely
the big items and items in $S_{\cross}^{\nice}$ and the items
partially overlapping with the removed strip in \Cref{strip-removal}.

The main idea is to fix some items in the knapsack and then carry out our
algorithm discussed in the previous sections with some modifications. Again,
we may leave out some items with high profit. We fix these items too and repeat
this process. We can argue that at some point, the left out items have a very
less profit and hence this process will end.

Let us define the set $K(0)$ as the set of items that were removed in the above
process. If the total profit of $K(0)$ is at most $\epsilon\cdot \itemprofit{S}$, then we
are done. If this is not the case, then we use shifting argumentation.

Assume that we completed the $t\Th$ iteration for some $t>0$ and say
$\itemprofit{K(r)}>\epsilon\cdot \itemprofit{S}$ for all $0\le r\le t$. Let $\mathcal{K}(t) = \bigcup_{r=0}^t K(t)$.
We will argue that for every $r\le t$, $\abs{K(r)}$ is at most a constant. Then
$\abs{\mathcal{K}(t)}$ is also at most a constant if $t<\ceil{1/\epsilon}$
(in fact, we will argue that $t$ will not go beyond $\floor{1/\epsilon}$).

Let us define a non-uniform grid $\mathcal{G}(t)$ induced by the set $\mathcal{K}(t)$ as follows:
The $x$ and $y$ coordinates of the grid are given by all the corners of the items
in $\mathcal{K}(t)$. Note that the number of horizontal (resp. vertical) grid lines
is bounded by $2\abs{\mathcal{K}(t)}$. This grid partitions the knapsack into
a set of cells $\mathcal{C}(t)$. Since $\abs{\mathcal{K}(t)}$ is at most a constant,
the number of grid cells created is also at most a constant.

Let us denote an arbitrary grid cell by $C$ and the items in $S$ which
intersect $C$ and which are not in $\mathcal{K}(t)$ by $S(C)$.
For an item $i$ in $S(C)$,
let $\itemheight{i\cap C}$ and $\itemlength{i\cap C}$
denote the width and height of the overlap of an item
$i$ with $C$. We categorize the items in $S(C)$ relative to $C$ as follows.
\begin{itemize}
    \item $S_{\mini}(C)=\left\{i\in S(C): \itemheight{i\cap C}\le \epsilon_{\mini} \itemheight{C} \andtext \itemlength{i\cap C}\le \epsilon_{\mini} \itemlength{C}\right\}$
    \item $S_{\mammoth}(C)=\left\{i\in S(C): \itemheight{i\cap C}> \epsilon_{\mammoth} \itemheight{C} \andtext \itemlength{i\cap C}> \epsilon_{\mammoth} \itemlength{C}\right\}$
    \item $S_{\hor}(C)=\left\{i\in S(C): \itemheight{i\cap C}\le \epsilon_{\mini} \itemheight{C} \andtext \itemlength{i\cap C}> \epsilon_{\mammoth} \itemlength{C}\right\}$
    \item $S_{\ver}(C)=\left\{i\in S(C): \itemheight{i\cap C}> \epsilon_{\mammoth} \itemheight{C} \andtext \itemlength{i\cap C}\le \epsilon_{\mini} \itemlength{C}\right\}$
\end{itemize}
We call an item $i$ \emph{small} if it is not in $S_{\mammoth}(C)\cup S_{\hor}(C)\cup S_{\ver}(C)$ for some cell $C$.
We call an item $i$ \emph{big} if it is in $S_{\mammoth}(C)$ for some cell $C$.
We call an item $i$ \emph{wide} (resp. \emph{tall}) if it is in $S_{\hor}(C)$ (resp. $S_{\ver}(C)$) for some cell $C$.

We call an item $i$ \emph{medium} if there is a cell $C$ such that
$\itemheight{i\cap C}\in\left(\epsilon_{\mini} \itemheight{C}, \epsilon_{\mammoth} \itemheight{C}\right]$ or
$\itemlength{i\cap C}\in\left(\epsilon_{\mini} \itemlength{C}, \epsilon_{\mammoth} \itemlength{C}\right]$.

It is easy to observe that an item must belong to exactly one of small, big,
wide, tall, medium categories. Note that the width and height
of a small item are at most $\epsilon_{\mini}$. We call an item $i$ skewed if
it is either wide or tall.

Again by standard arguments, we can select $\eps_\mini$ and $\eps_\mammoth$
such that the profit of the medium items is at most $\eps\cdot p(S)$.
Hence, we remove all the medium items. We also remove all the small items for
now, but we will pack them at the end exactly as in \Cref{pack-small}.
We add all the big items to $K(t+1)$. We can do this because the big
items are at most constant in number: Consider any cell $C$. The number
of big items associated with it is at most a constant and the number of cells
themselves is at most a constant.

We create a corridor decomposition in the knapsack with respect to
the skewed items as follows: First we transform this non-uniform grid into
a uniform grid by moving the grid lines and simultaneously stretching or compressing
the items. This transformation is to ensure that every wide (resp. tall) item
associated with a cell $C$ has at least a width (resp. height) of
$\epsilon_{\mammoth}/(1+2\abs{\mathcal{K}(t)})$. We apply \Cref{corr-packing}
on the set of skewed items into a set of $\bigOhConst{1/\epsilon,1/\epsilon_{\mammoth}}$
corridors. Let $S_{\corr}, S_{\cross}^{\nice}, S_{\cross}^{\bad}$ be the sets obtained
similar to \Cref{corr-packing}. The set $S_\corr$ is the set of items that are
packed in the corridors completely. We add the set $S_{\cross}^{\nice}$ to $K(t+1)$
since they are constant in number. As the set $S_{\cross}^{\bad}$ has a very small
profit, we discard them. We also remove all the items in tall subcorridors
assuming, \wLoG, that their profit is at most that of the items in wide
subcorridors.

Now, we have the items in $\mathcal{K}(t)$ fixed inside the knapsack and the
items in boxes (obtained after deleting the tall items and hence deleting the
tall subcorridors). We would like to split the boxes
into containers as in \Cref{box-to-cont} but there is an issue:
There can be items in $\mathcal{K}(t)$ which
overlap with the boundaries of the boxes. But these are at most a constant
in number and hence we can resolve this issue in the following way.

Consider an item $i$ in $\mathcal{K}(t)$ partially overlapping with a box.
We extend each of its horizontal edge that is inside the box in both
the directions till the ends of the box. This extended edge and $i$ divide the
items in the box into at most five categories: The items intersecting
the extended edge, the items to the left
of $i$, the items to the right of $i$, the items above $i$ and the items
below $i$. We note that the items
in the first part are at most a constant in number and hence add them to $K(t+1)$.
Thus $i$ splits the box into at most four smaller boxes. We repeat this process
for all the items in $\mathcal{K}(t)$ partially overlapping with the box. We obtain
smaller boxes but with the required property that there are no partially
overlapping items in $\mathcal{K}(t)$ with the boxes. This is depicted in \Cref{box-division}.

We apply \Cref{strip-removal} to these smaller boxes and while doing this, we add all the
items partially overlapping with the removed strip to $K(t+1)$.
Then, we apply \Cref{ra-lemma} to these new smaller boxes to
split them into containers as in \Cref{box-to-cont}. At this point, the $(t+1)\Th$
round ends and we look at the set $K(t+1)$. If $\itemprofit{K(t+1)}\le \epsilon\cdot \itemprofit{S}$,
we end; otherwise we continue to round $t+2$. We can argue that the number of rounds
are at most $1/\epsilon$: $K(r)$ and $K(r+1)$ are disjoint for all $r\ge 0$. Hence, for some
$r<\ceil{1/\epsilon}$, we can guarantee that $\itemprofit{K(r)}$ is at most $\epsilon\cdot \itemprofit{S}$.

Therefore, after the $r\Th$ round, we end the process and we add the small items
to the packing as described in \Cref{pack-small}.

In \cite{l-packing}, it is also shown how to modify the formed containers
so that their widths and heights come from a set of cardinality $\poly(\abs{S})$.
The basic idea is as follows. First note that a large container
contains only one item. Hence, the size of the container can be shrunk to that
of the item. Now, consider a wide container $C_W$. Items are lying one on top of
another inside $C$. So, $h(C_W)$ can be expressed as
the sum of heights of all the items. Also, $w(c_W)$ can be expressed as
the maximum among all the widths of items in the container.
If there are at most $1/\eps$ number of items in $C_W$, then $h(C_W)$
can belong to a set containing at most $n^{1/\eps}$ values. If $C_W$ has more
than $1/\eps$ number of items, then we consider the tallest $1/\eps$ number of
items in $C_W$. Let this set be $T$. There must certainly exist an item $i\in T$
whose profit is at most $\eps p(C_W)$. We remove the item $i$ from the
container and readjust the height of $C_W$ so that its height belongs to
a polynomial set of values. Now consider an area container $C_A$. Since all
items in $C_A$ are $\eps$-small with respect to the dimensions of $C_A$, we
can leave out a subset of items carrying a very marginal profit and create some
space in both the horizontal and vertical directions so that the width and height
of $C_A$ can be readjusted. For full proof one can refer to
\cite{l-packing}.
This proves \Cref{structresult}.
\begin{figure}
\begin{center}
\begin{tikzpicture}[scale=2.0]
\coordinate (A) at (0,0);
\coordinate (B) at (2,0);
\coordinate (C) at (2,1);
\coordinate (D) at (0,1);

\coordinate (X) at (0.8,0.35);
\coordinate (Y) at (1.2,0.35);
\coordinate (Z) at (1.2,0.5);
\coordinate (W) at (0.8,0.5);

\coordinate (U) at (0,0.5);
\coordinate (V) at (0,0.35);

\draw [thick] (A)--(B)--(C)--(D)--cycle;
\draw [fill] (X)--(Y)--(Z)--(W)--cycle;

\coordinate (E) at (intersection of A--D and U--Z);
\coordinate (F) at (intersection of B--C and U--Z);

\coordinate (G) at (intersection of A--D and V--X);
\coordinate (H) at (intersection of B--C and V--X);

\draw [dashed] (E)--(F);
\draw [dashed] (G)--(H);

\draw [fill,fill opacity=0.5] (1.3,0.45)--(1.7,0.45)--(1.7,0.55)--(1.3,0.55)--cycle;
\draw [fill,fill opacity=0.5] (0.2,0.45)--(0.6,0.45)--(0.6,0.55)--(0.2,0.55)--cycle;
\draw [fill,fill opacity=0.5] (1.3,0.32)--(1.7,0.32)--(1.7,0.37)--(1.3,0.37)--cycle;
\draw [fill,fill opacity=0.5] (0.2,0.32)--(0.6,0.32)--(0.6,0.37)--(0.2,0.37)--cycle;
\end{tikzpicture}

\caption{Division of box into smaller boxes: The dark item is the item in $\mathcal{K}(t)$
that overlaps with the box. The dashed lines are the extended edges and the grey items
are those that will be included in $K(t+1)$}\label{box-division}
\end{center}
\end{figure}
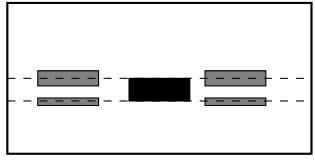

\subsection{Container Packing Problem and Vector-Max-GAP}
\label{cont-vecmaxgap}
From now on, our main goal would be to derive an algorithm to construct the
nice structure as in \Cref{structresult}. We first formally define the problem
of packing items into containers which we call the \emph{container packing problem}.
We then model the container packing problem as an instance of the \vmgap{} problem.
Since we have a PTAS for the \vmgap{} problem, it follows that there exists
a PTAS for the container packing problem.

Let $I$ be a set of items and let
$\lengthinp,\heightinp,\profitinp,\weightinp$ denote the associated widths, heights,
profits and weights respectively.

Let $C$ be a given set of containers such that the number of containers is
constant. Out of $C$, let $C_A, C_H, C_V, C_L$ denote area, wide, tall and large containers
respectively.

In the Container Packing Problem, we would like to pack a subset of
$I$ into these containers such that
\begin{itemize}
  \item A large container can either be empty or can contain exactly one item.
  \item In a wide (resp. tall) container, all the items must be
  stacked up one on top of another (resp. arranged side by side).
  \item The total area of items packed in an area container must not exceed
  $(1-\epsilon')^2$ times the area of the container itself (assume that
  $\epsilon'$ is a constant given as a part of the input).
  \item The total weight of items packed in all the containers in any dimension
  $j\in\{1,\dots,d\}$ should not exceed one.
\end{itemize}
We denote an instance of the container packing problem by the tuple
$\containerinstance{I}{C}{\epsilon'}$.

Now, let us define the Vector-Max-GAP problem. Let $I$ be a set of $n$ items and let
us say, we have $k$ machines such that the $j\Th$ machine has a capacity $M_j$. An item
$i$ has a size of $s_j(i)$, value of $\val_j(i)$ in the $j\Th$ machine and weight $w_q(i)$ in the
$q\Th$ dimension. The objective is to assign a maximum value
subset of items $I'\subseteq I$, each item to a machine, such that the total size of
items in a machine does not exceed the capacity of that machine. We also require
that the total weight of $I'$ does not exceed $W_q$ (a non-negative real) in any dimension $q\in[d]$.

Let $\vect{M} = [M_1,M_2,\dots,M_k]$, $\vect{W} = [W_1,W_2,\dots,W_d]$, $\vect{w}(i) = [w_1(i),w_2(i),\dots,w_d(i)]$, $\vect{s}(i) = [s_1(i),s_2(i),\dots,s_k(i)]$, $\vect{\val}(i) = [\val_1(i),\val_2(i),\dots,\val_k(i)]$.

Also, let $\vect{s} = [\vect{s}(1),\vect{s}(2),\dots,\vect{s}(n)]$, $\vect{w} = [\vect{w}(1),\vect{w}(2),\dots,\vect{w}(n)], \vect{\val} = [\vect\val(1),\vect\val(2),\dots,\vect\val(n)]$. We denote
an instance of Vector-Max-GAP by $(I, \vect{\val}, \vect{s}, \vect{w}, \vect{M}, \vect{W})$.

\subsection{Reduction of Container Packing Problem to Vector-Max-GAP}
\label{red-cont-vecmaxgap}
Let $\mathcal{C} = \containerinstance{I}{C}{\epsilon'}$ denote an instance of the container packing
problem. In this subsection, we show how to reduce $\mathcal{C}$ to an instance
$\mathcal{G} = (I, \overrightarrow{\val}, \overrightarrow{s}, \overrightarrow{w}, \overrightarrow{M}, \vect{W})$ of the Vector-Max-GAP problem.

We retain $I$ from $\mathcal{C}$. Since we have unit vector
constraints over all the containers combined in the container packing problem,
we initialize $\vect{W}$ to be the vector of all ones.
We choose the number of machines in $\mathcal{G}$ to be same as the number of containers in
$\mathcal{C}$. Let $\abs{C}=k$ and $\abs{I}=n$. The vectors $\overrightarrow{\val}, \overrightarrow{s}, \overrightarrow{M}$
are defined as follows:
\begin{align*}
s_j(i) &= \begin{cases}
    1 & \textup{if $C_j$ is a large container and $i$ can fit in inside $C_j$}\\
    \infty & \textup{if $i$ can not fit in inside $C_j$}\\
    \itemheight{i} & \textup{if $C_j$ is a wide container and $i$ fits in inside $C_j$}\\
    \itemlength{i} & \textup{if $C_j$ is a tall container and $i$ fits in inside $C_j$}\\
    \itemlength{i}\itemheight{i} & \textup{if $C_j$ is an area container
        and $i$ is $\epsilon'$-small for $C_j$}\\
    \infty & \textup{if $C_j$ is an area container but $i$ is not $\epsilon'$-small for $C_j$}\\
\end{cases} \\
M_j &= \begin{cases}
    1 & \textup{if $C_j$ is a large container}\\
    \itemheight{C_j} & \textup{if $C_j$ is a wide container}\\
    \itemlength{C_j} & \textup{if $C_j$ is a tall container}\\
    (1-\epsilon')^2\itemheight{C_j}\itemlength{C_j} & \textup{if $C_j$ is an area container}\\
\end{cases} \\
\val_j(i) &= \itemprofit{i} \\
w_q(i) &= v_q(i)
\end{align*}
In the above definitions of $s_j(i), M_j, \val_j(i)$ and $w_q(i)$,
$i$ varies from 1 to $n$, $j$ varies from 1 to $k$ and
$q$ varies from 1 to $d$.

Let $I'$ denote a subset of $I$ packed in a feasible solution to $\mathcal{C}$.
Then $I'$ can also be feasibly assigned to the machines of our Vector-Max-GAP problem:
Just assign all the items in a container $C_j$ to the $j\Th$ machine.

\begin{itemize}
\item If $C_j$ is a large container, then the only item packed into it has
size 1 in machine $M_j$ and capacity of $M_j$ is also 1 and hence assigning this
item to the $j\Th$ machine is feasible.
\item If $C_j$ is a wide (resp. tall) container, the items packed in
$C_j$ are wide and stacked up (resp. tall and arranged side by side).
Hence their total height (resp. width), which is the total size of items assigned
to the $j\Th$ machine, does not exceed the total height (resp. width) of the
container, which is the capacity of the $j\Th$ machine.
\item If $C_j$ is an area container, the total area of items packed in $B_j$
does not exceed $(1-\epsilon')^2\cdot \areasym(C_j)$ which yields that the total
size of items assigned to the $j\Th$ machine does not exceed the
capacity of the $j\Th$ machine.
\end{itemize}

The total weight of all the items assigned to machines does not exceed
$\overrightarrow{W}$ (which is equal to the all ones vector) as we did not
change $\weightinp$ while reducing $\mathcal{C}$ to $\mathcal{G}$.
This proves that the optimal value obtainable for $\mathcal{G}$ is at least
as much as that of the container packing problem $\mathcal{C}$.

On the other hand, consider any feasible assignment of a subset of items $J\subseteq I$
to the machines in our instance $\mathcal{G}$. Let $J_j$ be the subset
of items assigned to the $j\Th$ machine. We can pack $J_j$ into container
$C_j$ in the following way: Since the assignment is feasible, the size of all
items in $J_j$ does not exceed the capacity of $M_j$. If $C_j$ is wide
(resp. tall), $\sum_{i\in J_j}\itemheight{i} \le \itemheight{C_j}$
(resp. $\sum_{i\in J_j}\itemlength{i} \le \itemlength{C_j}$).
Hence, all the items in $J_j$ can be stacked up (resp. arranged side by side)
inside the container $C_j$. If $C_j$ is an area
container, then $J_j$ consists of only small items which
are $\epsilon'$-small for $C_j$ and $\areasym(J_j) \le (1-\epsilon')^2\cdot \areasym(C_j)$.
Hence, by \cref{lem:nfdh-small}, we can pack the set
$J_j$ into $C_j$ using NFDH. If an item is assigned
to a large container, then it occupies the whole capacity (since item size and
machine capacity are both equal to 1) and hence, only a single item can be
packed in a large container.

The above arguments prove that the container packing problem is a special
instance of the Vector-Max-GAP problem. Hence, we can use the PTAS for the
\vmgap{} problem from \cref{vec-max-gap} to obtain a PTAS for the container
packing problem.

The reduction in the case when rotations are allowed is exactly the same except for the
values of $s_j(i)$ in the reduction of Container Packing Problem to an instance
of the Vector-Max-GAP problem: If $C_j$ is a tall (resp. wide)
container,
\begin{align*}
s_j(i) &= \begin{cases}
            \infty& \textup{if $i$ can fit neither with rotation nor without rotation}\\
            \itemlength{i}\:(\mathrm{resp.}\:\itemheight{i})& \textup{if $i$ fits without rotation but not with rotation}\\
            \itemheight{i}\:(\mathrm{resp.}\:\itemlength{i})&\textup{if $i$ fits with rotation but not without rotation}\\
            \min\{\itemlength{i},\:\itemheight{i}\}& \textup{if $i$ fits both with and without rotation}\\
          \end{cases}
\end{align*}
If $C_j$ is a large container, we set $s_j(i)=\infty$ if $i$ does not fit in $C_j$ with or
without rotations. Otherwise we set $s_j(i)$ to 1. In case of $C_j$ being an area container,
$s_j(i)$ is same as the case without rotations. The correctness of the reduction
follows by similar arguments as above.
\begin{theorem}
\label{ptascontainer}
There exists a PTAS for the container packing problem with or without rotations.
\end{theorem}

\subsection{Algorithm}
\label{mixknap-algo}
Our main goal is to search for the container packing structure in \Cref{structresult}.

For this we need to guess the set of containers used to pack the items in $I'_{\opt}$ of
(\ref{halfprofit}). As mentioned in \cref{geomvecks-struct}, the number of containers
used is at most a constant (let this number be denoted by $c$) and they have to be
picked from a set whose cardinality is in $\poly(\abs{I})$
(let this cardinality be denoted by $q(\abs{I})$). Therefore, the number of guesses
required is $\binom{q(\abs{I})+c}{c}$ which is in turn in $\poly(\abs{I})$.

Once we have the containers, we need to guess which of them are large containers, wide
containers, tall containers and area containers. The number of guesses required
is at most $\binom{c+4}{c}$ which is a constant.

Then we use the PTAS of the container packing problem with approximation factor
$(1-\epsilon_{\mathrm{cont}})$, and since the optimal profit of the container packing
problem is at least $(1/2-\epsilon_\struct)\cdot \optgvks(I)$, by choosing
$\epsilon_{\mathrm{cont}}:=\epsilon$ and $\epsilon_{\struct}:=\epsilon/2$, we
get the desired approximation factor of $(1/2-\epsilon)$ for the \geomvec{2}{$d$} knapsack
problem.

\section{Conclusion}
We study the $(2,d)$ Knapsack problem and design a $(2+\epsilon)$ approximation
algorithm using corridor decomposition and the PTAS for
\vmgap{} problem, which could be of independent interest.
We believe that the approximation ratio might be improved to
$1.89$, however, it will require involved applications of complicated techniques like
$L$-packings as described in \cite{l-packing}.
One might also be interested in studying a further generalization of the problem, the
$(d_G,d_V)$ Knapsack where items are $d_G$ dimensional hyper-rectangles
with weights in $d_V$ dimensions.
One major application of our result will be in designing approximation algorithms for
$(2,d)$ Bin Packing problem, a problem of high practical relevance.
In \cite{mixed-bin-packing}, we use the presented approximation algorithm for $(2,d)$ Knapsack as a subroutine
to obtain an approximation algorithm for the $(2,d)$ Bin Packing problem.

\bibliographystyle{plainurl}
\bibliography{bibdb}
\newpage
\appendix
\section{Appendix}

\subsection{Next Fit Decreasing Heights (NFDH)}
NFDH, introduced in \cite{coffman1980performance},
is a widely used algorithm to pack rectangles in a bigger rectangular bin. It works
as follows. First, we order the rectangles in the decreasing order of their heights.
We then place the rectangles greedily on the base of the bin until we do not cross
the boundaries of the bin. At this point, we ``close'' the shelf and shift the
base to top edge of the first rectangle and continue the process. For an
illustration of the algorithm, see \cref{fig:nfdh}.
\begin{figure}[H]
\centering
\begin{tikzpicture}[scale=3]
\draw (0,0) rectangle (1,1);
\draw[fill=black!30] (0,0) rectangle ++(0.3,0.5);
\draw[fill=black!30] (0.3,0) rectangle ++(0.5,0.4);
\draw[fill=black!30] (0,0.5) rectangle ++(0.5,0.3);
\draw[fill=black!30] (0.5,0.5) rectangle ++(0.1,0.2);
\draw[fill=black!30] (0.6,0.5) rectangle ++(0.4,0.15);
\end{tikzpicture}
\caption{The NFDH Algorithm. After packing the first two items on the base of the
bin, the third item can't be packed on the same level. Hence, we close the shelf
and create a new shelf.}
\label{fig:nfdh}
\end{figure}
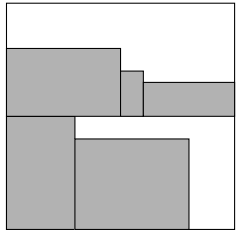
Now, we state a well-known and important lemma about NFDH. 

\begin{lemma}\cite{bansal2009structural}.
Let $S$ be a set of rectangles and let $w$ and $h$ denote the largest width
and largest height in $S$, respectively. Consider a bin of dimensions $W\times H$. If
$\itemarea{S}\le (W-w)(H-h)$, then NFDH packs all the rectangles into the bin.
\end{lemma}

The above lemma suggests that NFDH works very well if all the rectangles are small 
compared to the biun dimensions. 
The following lemma will be crucial for our purposes.
\begin{lemma}
\label{lem:nfdh-small}
Consider a set of rectangles $S$ with associated profits
and a bin of dimensions $W\times H$ and assume that
that each item in $S$ has width at most $\epsilon W$ and height at most $\epsilon H$.
Suppose $\opt$ denotes the optimal profit that can be packed in the bin. Then there
exists a polynomial time algorithm that packs at least $(1-2\epsilon)\opt$
profit into the bin.
\end{lemma}
\begin{proof}
Let us first order the rectangles in the decreasing order of their profit/area
ratio. Then pick the largest prefix of rectangles $T$ such that
$\itemarea{T}\le (1-\epsilon)^2WH$. By the above lemma, NFDH must be able to pack
all the items in $T$ into the bin. On the other hand, since each rectangle has area
at most $\epsilon^2WH$, it follows that $\itemarea{T}\ge(1-2\epsilon)WH$.
Furthermore, since $T$ contains the highest profit density items, it follows that
$\itemprofit{T}\ge(1-2\epsilon)\opt$.
\end{proof}
\end{document}